\documentclass[draft, fontsize=11pt]{scrartcl}
\usepackage{amsmath,amssymb,amsthm}
\usepackage[longnamesfirst,numbers]{natbib}
\usepackage[margin]{fixme}
\usepackage{setspace}
\usepackage{enumerate}         

\newtheorem{thm}{Theorem}[section]
\newtheorem{cor}[thm]{Corollary}
\newtheorem{lem}[thm]{Lemma}

\theoremstyle{definition}

\newtheorem{ass}[thm]{Assumption}
\theoremstyle{remark}
\newtheorem{rem}[thm]{Remark}

\newtheorem{counterexample}[thm]{Counterexample}
\numberwithin{equation}{section}

\newcommand{\set}[1]{\left\{#1\right\}}
\newcommand{\Ind}[1]{\mathbf{1}_{\left\{#1\right\}}}
\newcommand{\RR}{\mathbb{R}}

\newcommand{\PP}{\mathbb{P}}

\newcommand{\NN}{\mathbb{N}}

\newcommand{\FF}{\mathbb{F}}
\newcommand{\HH}{\mathbb{H}}

\newcommand{\cF}{\mathcal{F}}
\newcommand{\cG}{\mathcal{G}}
\newcommand{\cH}{\mathcal{H}}

\newcommand{\cP}{\mathcal{P}}

\newcommand{\cR}{\mathcal{R}}

\newcommand{\Rplus}{\mathbb{R}_{\geqslant 0}}

\newcommand{\scal}[2]{\left\langle{#1},{#2}\right\rangle}

\newcommand{\wt}[1]{{\widetilde{#1}}}
\newcommand{\Econd}[2]{\mathbb{E}\left[\left.#1\right|#2\right]}        
\newcommand{\E}[1]{\mathbb{E}\left[#1\right]}                     

\newcommand{\wh}[1]{{\widehat{#1}}}

\begin{document}
\title{\LARGE Convex order of discrete, continuous and predictable quadratic variation \&  applications to  options on variance\thanks{The authors would like to thank Johannes Muhle-Karbe, Mark Podolskij, David Hobson and Walter Schachermayer for discussions and comments.}}
\author{\large Martin Keller-Ressel\\ \large TU Berlin, Institut f\"ur Mathematik\\ \large mkeller@math.tu-berlin.de \and \large Claus Griessler\\ \large Universit\"at Wien, Fakult\"at f\"ur Mathematik\\ \large claus.griessler@univie.ac.at}

\date{\large September~25, 2012}

\maketitle
\begin{abstract}
We consider a square-integrable semimartingale and investigate the convex order relations between its discrete, continuous and predictable quadratic variation. As the main results, we show that if the semimartingale has conditionally independent increments and symmetric jump measure, then its discrete realized variance dominates its quadratic variation in increasing convex order. The results have immediate applications to the pricing of options on realized variance. For a class of models including time-changed L\'evy models and Sato processes with symmetric jumps our results show that options on variance are typically underpriced, if quadratic variation is substituted for the discretely sampled realized variance.
\end{abstract}


\section{Introduction}
\subsection{Variance Options}
Let $X$ be a stochastic process, and let $\cP$ be a partition of $[0,T]$ with $n + 1$ division points
\[0 = t^n_0 \le t^n_1 \le \dotsm \le t^n_n = T.\]
Then the \emph{realized variance} of $X$ over $\cP$ is given by
\begin{equation}\label{Eq:RV}
RV(X,\cP) = \sum_{k=1}^n \left(X_{t_{k}} - X_{t_{k-1}}\right)^2.
\end{equation}
In a financial market, let $S$ be the price of a stock, a currency
future, or another security, and let $X_t = \log(S_t/S_0)$ be its
(normalized) logarithm. Then the realized variance \eqref{Eq:RV} of
$X$, with the increments $t_k - t_{k-1}$ typically being business
days, is a way to measure the volatility of the security price over
the time interval $[0,T]$. Often realized variance is scaled by
$\frac{1}{T}$ such that it measures `annualized variance' over a
time interval of unit length. Since the uncertainty of future
volatility is a risk factor to which all market participants are
exposed, many wish to hedge against it, while others are willing to
take on volatility risk against a premium. To this end
\emph{volatility derivatives} have emerged, that allow market
participants to take up positions in realized variance, see e.g.
\citet{Carr2009} for an overview. Many of these derivatives pay, at
the terminal time $T$, an amount $f(\tfrac{1}{T}RV(X,\cP))$ to the
holder, where $f$ is the payoff function\footnote{Other examples are
so-called weighted variance swaps, which are volatility derivatives,
but whose payoff cannot be written as $f(\tfrac{1}{T}RV(X,\cP))$;
see \citet{Lee2010a}.}. We call such derivatives \emph{options on
variance}. Typical choices for the payoff function are
\begin{quote}
\begin{tabular}{l@{\hspace{30pt}}l}
\textbf{variance swap}: $f(x) = x - K$& \textbf{volatility swap}: $f(x) = \sqrt{x} - K$,\\[10pt]
\textbf{call option}: $f(x) = (x - K)^+$& \textbf{put option}: $f(x) = (K - x)^+$
\end{tabular}
\end{quote}
where $K \in \Rplus$. For the variance swap, $K$ may be chosen in such a way that today's fair value of the swap is zero; this choice of $K$ is called the swap rate, and we denote it by $s$. The strike $K$ for e.g. call options can then be chosen relative to the swap rate $s$ by setting $K = ks$ for some $k \in \Rplus$. We refer to this choice as \emph{relative strike} options. Note that the payoff functions listed above, with the exception of the volatility swap, are convex. In this article we will be concerned with variance options defined through a generic convex payoff function $f: \Rplus \to \RR$.
\subsection{Quadratic Variation and the Convex Order Conjecture(s)}
One of the cornerstones of the valuation theory of variance options
is the replication argument of \citet{Neuberger1992} (see also e.g.
\citet{Lee2010a}), which states that a variance swap can be
replicated by holding a static portfolio of co-maturing European options on $S$,
while dynamically trading in the underlying $S$. The strength of
Neuberger's replication argument lies in the fact that it is
essentially model-free, up to the assumptions that (a) $S$ follows a
continuous martingale, (b) European options of all
strikes and co-maturing with the variance swap are traded, and (c)
that the realized variance \eqref{Eq:RV} can be substituted by the
quadratic variation $[X,X]_T$ of $X$. Here, we are mainly interested in
the last assumption, which is based on the fact that for any
semimartingale $X$, and sequence of partitions $(\cP^n)_{n \in
\NN}$ of $[0,T]$, the realized variance $RV(X,\cP^n)$ converges in
probability to the quadratic variation $[X,X]_T$ as the mesh of the
partition tends to zero; i.e.
\begin{equation}\label{Eq:RV_convergence}
RV(X,\cP^n) \to [X,X]_T, \qquad \text{in probability}
\end{equation}
as $\mathrm{mesh}(\cP^n) \to 0$. For options on variance with
non-linear payoff, the static replication argument of Neuberger
breaks down, but in specific cases like the Heston model, dynamic
replication strategies involving European options or variance swaps
can be derived (see \citet{Broadie2008a}). In general, even if
perfect replication is not possible, the arbitrage-free price at
time zero of an option on variance is given by $e^{-rT}\E{f(\tfrac{1}{T}RV(X,\cP))}$ where $\E{.}$ denotes an
expectation under the risk-neutral pricing measure, and $r$ the
risk-free interest rate. Also for risk-neutral
pricing and imperfect hedging, realized variance is frequently
substituted by quadratic variation, since the latter is both
conceptually and computationally easier to use, and eliminates the
dependency on the nature of the partition $\cP$. See
\citep{Buhler2006, Carr2003, Kallsen2009} for examples of this
approach. All this raises questions on the qualities of the
approximation
\[\E{f\left(\frac{1}{T}RV(X,\cP)\right)} \approx \E{f\left(\frac{1}{T}[X,X]_T\right)}, \]
that is (a) how precise is it? and (b) is there a systematic bias? While the precision of the approximation has been studied in the asymptotic limit $n \to \infty$ and $T \to 0$ (see \citet{Broadie2008, Sepp2010, Farkas2010} resp. \citet{KM2010}), we are here interested in the existence of a systematic bias \emph{without asymptotics}, i.e. for fixed and finite $n$ and $T$. Numerical evidence given in \citet{Gatheral2008, Buhler2006, KM2010} strongly supports such a bias and suggests that the price of a variance option with convex payoff, evaluated on discretely sampled variance, is higher than the price of an option with the same payoff, evaluated on quadratic variation. From this evidence we are led to the conjecture that
\begin{equation}\label{Eq:convex_order_conj}
\E{f\left(RV(X,\cP)\right)} \ge \E{f\left([X,X]_T\right)},
\end{equation}
for all convex functions $f$, or in other words that realized
variance dominates quadratic variation in convex
order.\footnote{Note that the factor $\frac{1}{T}$ will be absorbed
into the function $f$ from now on.} We call this statement the
`convex order conjecture' between discrete and continuous variance and write it in more concise form as
$RV(X,\cP) \ge_\text{cx} [X,X]_T$.\footnote{This conjecture has been
initially conceived in discussions with Johannes Muhle-Karbe at the
Bachelier World Congress 2010 in Toronto.} A slightly weaker version
is obtained if \eqref{Eq:convex_order_conj} is required to hold only
for all convex, \emph{increasing} functions $f$. We call this the
`increasing convex order conjecture'; it is equivalent to the
realized variance dominating quadratic variation in increasing
convex order, or more concisely $RV(X,\cP) \ge_\text{icx} [X,X]_T$.
If $\E{RV(X,\cP)} = \E{[X,X]_T}$ then the convex order conjecture
and the increasing convex order conjecture are equivalent, see \citet[Thm.~4.A.35]{Shaked2007}. 

Similar questions can be asked about the relationship between the quadratic variation $[X,X]$ and the \emph{predictable} quadratic variation $\scal{X}{X}$. For continuous semimartingales, of course, $\scal{X}{X}$ coincides with $[X,X]$. For discontinuous processes $\scal{X}{X}$ is different from $[X,X]$, but sometimes more analytically tractable and has been used as a substitute for realized variance, e.g. in \citet{Kallsen2009} for exactly this reason. In the same article systematic underpricing of variance puts with predictable quadratic variation in comparison to quadratic variation has been observed, such that one may conjecture the relation
\begin{equation}\label{Eq:convex_order_conj_pred}
\E{f\left([X,X]_T\right)} \ge \E{f\left(\scal{X}{X}_T\right)}
\end{equation}
for all convex functions $f$. Note that $\scal{X}{X}$ is the predictable compensator of $[X,X]$ such that under suitable integrability assumptions $\E{[X,X]_T} = \E{\scal{X}{X}_T}$ and increasing convex order is equivalent to convex order. The main goal of this article is to prove the presented conjectures under certain assumptions on $X$ and to outline the consequences for the pricing of options on variance.

\subsection{Strategy of the proofs and related work}
The convex (or increasing convex) order relation $Y \le_\text{cx} Z$ ($Y \le_\text{icx} Z$) between two random
variables is a statement about the probability laws of $Y$ and $Z$,
and thus not sensitive to the nature of the dependency between $Y$
and $Z$. Nevertheless it is often a useful strategy to couple $Y$ and $Z$, i.e. to define them on a common probability space, which allows to use stronger and
more effective tools, typically martingale arguments.\footnote{See
\citet{Lindvall1992} for the many uses of couplings in probability
theory.} For the proof of the convex and increasing convex order conjecture between discrete and continuous quadratic variation, we will couple $RV(X,\cP)$ and $[X,X]_T$ by embedding them into a reverse
martingale, as the initial and the limit law respectively.

The idea for the construction of this reverse martingale comes from the literature on quadratic variation of L\'evy processes, or more generally processes with independent increments. \citet{Cogburn1961} show that the realized variance of a process with independent increments over a sequence of nested partitions $\cP^n$ converges almost surely (and not just in probability) to the quadratic variation. The crucial step of their proof is to show that the sequence of realized variances over nested partitions forms a \emph{reverse martingale} when the underlying process has symmetric distribution. The almost sure convergence then follows from Doob's martingale convergence theorem and can be extended to arbitrary processes with independent increments by a symmetrization argument. This technique can in fact be traced back to \citet{Levy1940} where a corresponding result for Brownian motion is shown. 

Finally let us remark, that convex order relations in the context of variance options have also been explored by \citet{Carr2010}, but with a very different objective. The authors
observe that in an exponential L\'evy model $S_t = S_0e^{X_t}$,
where $X$ is a L\'evy process without Gaussian component, the
annualized quadratic variation $\tfrac{1}{t}[X,X]_t$ is decreasing
with $t$ in convex order. Consequently, for each convex payoff $f$
the term structure $t \mapsto \E{f(\frac{1}{t}[X,X]_t)}$ of variance
options is decreasing in $t$. Also in the proof of \citet{Carr2010} a reverse martingale is used: the process $t \mapsto \frac{1}{t}[X,X]_t$.

\section{Definitions and Preliminaries}
We briefly introduce some definitions and properties that will be used in the results and proofs of Section~\ref{Sec:main}.\\

For a stochastic process $X$ and a partition $\cP$ of $[0,T]$, the realized variance $RV(X,\cP)$ of $X$ over $\cP$ has been defined in \eqref{Eq:RV}. In the case that $RV(X,\cP)$ has finite expectation, we also define the \emph{centered realized variance} by
\[\overline{RV}(X,\cP) = RV(X,\cP) - \E{RV(X,\cP)}.\]
Finally, to allow for certain generalizations of our results, the \emph{$h$-centered realized variance} is defined, for any function $h: \Rplus \to \RR$, by
\begin{equation}\label{Eq:centering}
\overline{RV}^h(X,\cP^n) = RV(X,\cP^n) - h\Big(\E{RV(X,\cP^n)}\Big).
\end{equation}
Realized variance and centered realized variance can be regarded as the special cases $h(x) = 0$ and $h(x) = x$. We use the same notation to define an $h$-centered version of the quadratic variation $[X,X]_T$, i.e. we define
\[\overline{[X,X]}^h_T = [X,X]_T - h\Big(\E{[X,X]_T}\Big),\]
provided the expectation is finite. We will only be interested in $h$-centerings where $h: \Rplus \to \RR$ is Lipschitz continuous with Lipschitz constant at most $1$; we denote the set of these functions by $\mathrm{Lip}_1(\Rplus)$.\\
A sequence $(\cP^n)_{n \in \NN}$ of partitions is called \emph{nested}, if all division points of $\cP^{n-1}$ are also division points of $\cP^n$. Note that by inserting intermediate partitions we can always consider $(\cP^n)$ as a subsequence of a sequence $(\wt{\cP}^n)$ of nested partitions, where each $\wt{\cP}^n$ has exactly $n+1$ partition points. For the results in this article it will be sufficient to consider partition sequences of this type. Finally, the mesh of a partition $\cP$ is defined as usual by $\mathrm{mesh}(\cP) = \sup_{k \in \set{1, \dots, n}} (t_{k} - t_{k-1})$. As pointed out above
$RV(X,\cP^n) \to [X,X]_T$ in probability, if $\mathrm{mesh}(\cP^n) \to 0$ for any semimartingale $X$. This results holds also for non-nested partitions, see e.g. \citet[Thm.~I.4.47]{Jacod1987}.\\
Finally, let $(\cG_n)$ be a decreasing sequence of $\sigma$-algebras, and let $X_n$ be a sequence of integrable random variables such that $X_n \in \cG_n$. Then $(X_n)$ is called a \emph{reverse martingale} with respect to $\cG_n$, if
\begin{equation}\label{Eq:reverse_mg_prop}
\Econd{X_{n-1}}{\cG_n} = X_n
\end{equation}
for all $n \in \NN$. Note that $(\cG_n)$ is \emph{decreasing}, and thus not a filtration.
Similarly, $(X_n)$ is called a reverse submartingale if \eqref{Eq:reverse_mg_prop} holds with `$\ge$' instead of `$=$'.  An important result is the following (cf. \citet[29.3.IV]{Loeve1963}): A reverse submartingale converges almost surely and also in $L^1$ to a limit $X_\infty$. Note that in contrast to the (forward) submartingale convergence theorem, no additional conditions on $(X_n)$ are needed. If we define the tail $\sigma$-algebra $\cG_\infty = \bigcap_{k=1}^n \cG_n$, then for any $n \in \NN$ the limit $X_\infty$ can be represented as
\[\Econd{X_n}{\cG_\infty} = X_\infty.\]

\section{Results on discrete and continuous quadratic variation}\label{Sec:main}
\subsection{Reverse martingales from realized variance}
Let $(\Omega, \cF, \FF, \PP)$ be a filtered probability space, where $\FF$ satisfies the usual conditions of right-continuity and $\PP$-completeness. Let $Y$ be an $\FF$-adapted cadlag process, and let $\cH$ be a $\PP$-complete $\sigma$-algebra such that $\cH \subset \cF_0$. We say that $Y$ is a process with $\cH$-conditionally independent increments if for all $0 \le s \le t$ the increment $Y_t - Y_s$ is independent of $\cF_s$, conditionally on $\cH$. This definition includes time-changed L\'evy processes, and additive processes in the sense of \citet{Sato1999}. The $\cH$-conditional independence is equivalent to the assertion that
\begin{equation}\label{Eq:cond_independence}
\Econd{f(Y_t - Y_s)Z}{\cH} = \Econd{f(Y_t - Y_s)}{\cH} \cdot \Econd{Z}{\cH}
\end{equation}
for all bounded measurable $f$ and bounded $\cF_s$-measurable random variables $Z$. See also \citet[Chapter~6]{Kallenberg1997} for results on conditional independence and \citet{Jacod1987} for processes with conditionally independent increments.

Moreover, we say that a process $Y$ has $\cH$-conditionally symmetric increments, if 
\begin{equation}\label{Eq:cond_symmetry}
\Econd{f(Y_t - Y_s)}{\cH} = \Econd{f(Y_s - Y_t)}{\cH}
\end{equation}
for all $t,s \ge 0$ and bounded measurable $f$. Alternatively we can use conditional charateristic functions to characterize conditional symmetry, i.e. a random variable $X$ is $\cH$-conditionally symmetric if and only if
\begin{equation}\label{eq:charf_equal}
\Econd{\exp\left(iu X\right)}{\cH} = \Econd{\exp\left(-iuX\right)}{\cH}, \qquad \forall\, u \in \RR.
\end{equation}
It will be helpful to know that a process with conditionally symmetric and independent increments\footnote{Here and in the following `conditionally symmetric and independent' is a shorthand for `conditionally symmetric and conditionally independent'} is a martingale up to an integrability assumption.
\begin{lem}\label{Lem:martingale}
Let $X$ be an $\FF$-adapted process satisfying $\E{|X_t|} < \infty$ for all $t \ge 0$. If $X$ has $\cH$-conditionally symmetric and independent increments, then it is a martingale.
\end{lem}
\begin{proof}Let $t \ge s \ge 0$ and let $Z$ be a bounded $\cF_s$-measureable random variable. Using first conditional independence and then conditional symmetry of increments we obtain
\begin{align*}
\E{(X_t - X_s)Z} &= \E{\Econd{(X_t - X_s)Z}{\cH}} = \E{\Econd{(X_t - X_s)}{\cH}\cdot \Econd{Z}{\cH}} = \\
&= \E{\Econd{(X_s - X_t)}{\cH}\cdot \Econd{Z}{\cH}} = \E{(X_s - X_t)Z}
\end{align*}
and conclude that $\E{(X_t - X_s)Z} = 0$. Since $Z$ was an arbitrary bounded $\cF_s$-measurable random variable it follows that 
$\Econd{(X_t - X_s)}{\cF_s} = 0$ and hence that $X$ is a martingale.
\end{proof}
The following Lemma establishes a reflection principle for processes whose increments are conditionally symmetric and independent.
\begin{lem}\label{Lem:symmetrization}
Let $Y$ be a process with $\cH$-conditionally symmetric and independent increments and let $t_* \ge 0$. Then the process $\wh{Y}$ defined by
 \begin{equation}\label{Eq:reflection}
 \wh{Y}_s = Y_{s \wedge t^*} - \left(Y_s - Y_{s \wedge t^*}\right)
 \end{equation}
is equal in law to $Y$, conditionally on $\cH$.\\
Moreover, $RV(\wh{Y},\cP) = RV(Y,\cP)$ for all partitions $\cP$ for which $t_*$ is a partition point.
\end{lem}

\begin{proof}
Showing that $\wh{Y}$ is equal in law to $Y$, conditionally on $\cH$, is equivalent to showing that for any sequence $t_0 \le t_1 \le \dotsm t_n$ and bounded measurable $f: \RR^n \to \RR$ it holds that 
\begin{equation}\label{eq:equality}
\Econd{f(Y_{t_1} - Y_{t_0}, \dotsc, Y_{t_n} - Y_{t_{n-1}})}{\cH} = \Econd{f(\wh{Y}_{t_1} - \wh{Y}_{t_0}, \dotsc, \wh{Y}_{t_n} - \wh{Y}_{t_{n-1}})}{\cH}.
\end{equation}
To show this equality it is sufficient to prove the equality of the $\cH$-conditional characteristic functions of each side. By inserting intermediate points we may assume without loss of generality that $t_* = t_m$ for some $m \in \set{0,\dotsc, n}$. Using the $\cH$-conditional independence \eqref{Eq:cond_independence} and the $\cH$-conditional symmetry of increments we obtain
\begin{align*}
&\Econd{\exp\left(i \sum_{k=1}^n u_k \left(\wh{Y}_{t_n} - \wh{Y}_{t_{n-1}} \right) \right)}{\cH} = \\
&\Econd{\exp\left(i \sum_{k=1}^m u_k \left(Y_{t_k} - Y_{t_{k-1}}\right) \right)}{\cH} \cdot \Econd{\exp\left(-i \sum_{k=m+1}^n u_k \left(Y_{t_k} - Y_{t_{k-1}}\right) \right)}{\cH} = \\
&\Econd{\exp\left(i \sum_{k=1}^m u_k \left(Y_{t_k} - Y_{t_{k-1}}\right) \right)}{\cH} \cdot \Econd{\exp\left(i\sum_{k=m+1}^n u_k \left(Y_{t_k} - Y_{t_{k-1}}\right) \right)}{\cH} = \\
&= \Econd{\exp\left(i \sum_{k=1}^n u_k \left(Y_{t_k} - Y_{t_{k-1}}\right) \right)}{\cH},
\end{align*}
which proves \eqref{eq:equality}. 
It remains to show that $RV(\wh{Y},\cP) = RV(Y,\cP)$ for partitions $\cP$ for which $t_*$ is a partition point. Denote the partition by $t_0 \le \dotsm \le t_n$ as before and assume that $t_*$ is the $m$-th partition point, i.e. $t_* = t_m$. Then
\begin{align*}
RV(\wh{Y},\cP) &= \sum_{k=1}^m \left(Y_{t_k} - Y_{t_{k-1}}\right)^2 + \sum_{k=m+1}^n (-1)^2 \left(Y_{t_k} - Y_{t_{k-1}}\right)^2 = \\ &=\sum_{k=1}^n \left(Y_{t_k} - Y_{t_{k-1}}\right)^2 = RV(Y,\cP).\qedhere
\end{align*}
\end{proof}

Using this Lemma we show that we can construct a reverse martingale from the realized variances of a process with conditionally symmetric and  independent increments, when the realized variances are taken over nested partitions. For a process with (unconditionally) symmetric and independent increments this result has been shown by \citet[Thm.~1]{Cogburn1961}. Our result is a minor variation of the theorem of \citeauthor{Cogburn1961}, but will be generalized in the next section.

\begin{thm}\label{Thm:reverse_martingale}
 Let $X$ be a process with $\cH$-conditionally symmetric and independent increments and let $(\cP^n)_{n \in \NN}$ be a sequence of nested partitions of $[0,T]$ such that \linebreak $\E{RV(X,\cP^{n})} < \infty$ for all $n \in \NN$. Then the realized variance of $X$ over the partitions $(\cP^n)$ is a reverse martingale, i.e.
 it satisfies
 \begin{equation}\label{Eq:bw_martingale}
 \Econd{RV(X,\cP^{n-1})}{\cG_n} = RV(X,\cP^{n}), \qquad (n \in \NN)
 \end{equation}
 where
 \begin{equation}\label{Eq:define_Gn}
 \cG_n = \cH \vee \sigma\Big(RV(X,\cP^{n}), RV(X,\cP^{n+1}), \dotsc\Big).
 \end{equation}
\end{thm}

\begin{proof}
Since $(\cP^n)$ is a nested sequence of partitions, we may assume without loss of generality that $\cP^{n-1}$ and $\cP^n$ differ  by a single division point, which we denote by $t^*$. Denote by $a, b$ the two closest division points of $\cP^{n-1}$, i.e. $t^*$ is inserted into the interval $[a,b]$. Let $\wh{X}$ be the process $X$ reflected to the right of $t^*$ as in \eqref{Eq:reflection}. By Lemma~\ref{Lem:symmetrization} $\wh{X}$ is equal in law to $X$, conditionally on $\cH$. Moreover, also by Lemma~\ref{Lem:symmetrization}, $RV(\wh{X},\cP^k) = RV(X,\cP^k)$ for all $k \ge n$. This implies that
\begin{align}\label{Eq:equal_filtration}
\cG_n &= \cH \vee \sigma\Big(RV(X,\cP^{n}), RV(X,\cP^{n+1}), \dotsc\Big) = \\
      &= \cH \vee \sigma\Big(RV(\wh{X},\cP^{n}), RV(\wh{X},\cP^{n+1}),\dotsc\Big). \notag
\end{align}
To ease notation we abbreviate $\cR_n := \sigma\Big(RV(X,\cP^{n}), RV(X,\cP^{n+1}), \dotsc\Big)$. The next step is to calculate $U := \Econd{RV(X,\cP^{n-1}) - RV(X,\cP^{n})}{\cG_n}$. By direct calculation we obtain that 
 \begin{equation}\label{Eq:RV_decomp}
 RV(X,\cP^{n-1}) - RV(X,\cP^{n}) = (X_b - X_{t^*}) (X_{t^*} - X_a)
 \end{equation}
and conclude from the properties of conditional expectations that
\[\E{U H Q} = \E{(X_b - X_{t^*}) (X_{t^*} - X_a)HQ}, \qquad \forall\,H\in \cH, \,Q \in \cR_n.\]
Using Lemma~\ref{Lem:symmetrization} and \eqref{Eq:equal_filtration} we obtain
\begin{multline}
\E{(X_b - X_{t^*}) (X_{t^*} - X_a) HQ} = \E{\Econd{(X_b - X_{t^*}) (X_{t^*} - X_a) Q}{\cH}H} = \\
= \E{\Econd{(\wh{X}_b - \wh{X}_{t^*}) (\wh{X}_{t^*} - \wh{X}_a) Q}{\cH}H} = -\E{(X_b - X_{t^*}) (X_{t^*} - X_a) HQ},
\end{multline}
and hence that $\E{UHQ} = 0$ for all $H \in \cH, Q \in \cR_n$. From $\cG_n = \cH \vee \cR_n$ we conclude that indeed
\[\Econd{RV(X,\cP^{n-1}) - RV(X,\cP^{n})}{\cG_n} = U = 0\]
showing the reverse martingale property \eqref{Eq:bw_martingale}. \end{proof}

\subsection{Coupling Realized Variance and Quadratic Variation}
To introduce the quadratic variation to the setting we now add the assumption that $X$ is a semimartingale, such that $RV(X,\cP^n) \to [X,X]_T$ in probability whenever $\mathrm{mesh}(\cP^n) \to 0$. Furthermore we require that $X$ is a square-integrable semimartingale, i.e. that it is a special semimartingale with canonical decomposition $X = X_0 + N + A$ where $N$ is a square-integrable martingale and $A$ has square-integrable total variation. This assumption implies in particular that $\sup_{t \in [0,T]} (X_t - X_0)^2$ is integrable, see \citet[Ch.~V.2]{Protter2004}.\footnote{These assumptions are very natural in the context of variance options, where $X$ is the log-price of a security $S = S_0 e^X$. Under the pricing measure $e^X$ is a local martingale and the semimartingale property of $X$ follows automatically by \citep[I.4.57]{Jacod1987}. The square-integrability of $X$ then merely ensures that the prices of variance options (in particular variance swaps) are finite.} Combining these assumptions with Theorem~\ref{Thm:reverse_martingale} immediately gives the following Corollary.

\begin{cor}
Suppose that $X$ is a square-integrable martingale with $\cH$-conditionally symmetric and independent increments and that $(\cP^n)_{n \in \NN}$ is a nested sequence of partitions such that $\mathrm{mesh}(\cP^n) \to 0$. Then $RV(X,\cP^{n}) \to [X,X]_T$ a.s., $\E{[X,X]_T} < \infty$ and
 \begin{equation}\label{Eq:conditional_qv}
 \Econd{RV(X,\cP^n)}{\cG_\infty} = [X,X]_T
 \end{equation}
holds for any $\cP^n$ and with $\cG_\infty = \cH \vee \bigcap_{n \in \NN} \sigma(RV(X,\cP^n))$.
\end{cor}
\begin{proof}
Since $RV(X,\cP^n) \le 2(n+1)\sup_{t \in [0,T]} (X_t - X_0)^2$ and $X$ is square-integrable, we have that $\E{RV(X,\cP^n)} < \infty$ for all $n \in \NN$. It follows from Theorem~\ref{Thm:reverse_martingale} that the sequence $(RV(X,\cP^n))_{n \in \NN}$ is a reverse $\cG_n$-martingale. By the convergence theorem for reverse martingales this sequence converges almost surely and also in $L^1$, and it remains to identify the limit. Since $\mathrm{mesh}(\cP^n) \to 0$ we have by \citep[Theorem~I.4.47]{Jacod1987} that $RV(X,\cP^{n})$ converges in probability to the quadratic variation $[X,X]_T$. We conclude that $[X,X]_T$ is in fact the almost sure limit of $RV(X,\cP^{n})$ as $n \to \infty$, that $\E{[X,X]_T} < \infty$ and that $\Econd{RV(X,\cP^n)}{\cG_\infty} = [X,X]_T$.
\end{proof}

In the case that we are only interested in linear expectations of realized variance and quadratic variation, the symmetry assumption on the increments of $X$ can be dropped and we obtain the following.

 \begin{cor}\label{cor:expectations} Let $X$ be a square-integrable semimartingale with $\cH$-conditionally independent increments and let $\cP: 0 = t_0 \le t_1 \le \dotsc \le t_N = T$ be a partition of $[0,T]$. Then
 \begin{equation}\label{Eq:RV_expectation}
  \Econd{RV(X,\cP)}{\cH} = \Econd{[X,X]_T}{\cH} + \sum_{k=1}^{n}{\Econd{X_{t_{k}} - X_{t_{k-1}}}{\cH}^2}.
 \end{equation}
 In particular, $\Econd{RV(X,\cP)}{\cH} = \Econd{[X,X]_T}{\cH}$ for all partitions $\cP$ of $[0,T]$, if and only if the process $X$ is a martingale on $[0,T]$.
 \end{cor}
\begin{proof}
 Let $Y$ be an $\cH$-conditionally independent copy\footnote{Such a copy can be constructed using disintegration with respect to $\cH$.} of $X$. Then $Z = X - Y$ is a process with $\cH$-conditionally symmetric and independent increments. Decomposing $RV(X-Y,\cP)$ we get
\begin{align}\label{Eq:RV_decomp3}
 RV(X-Y,\cP) &= RV(X,\cP) + RV(Y,\cP) - \\
 &- 2\sum_{k=0}^{n-1}\left(X_{t_{k+1}} - X_{t_k}\right)\left(Y_{t_{k+1}} - Y_{t_k}\right), \notag
\end{align}
which can be bounded from above by $2 (RV(X,\cP) + RV(Y,\cP))$. Since $X$ is square-integrable, we have that $\E{RV(X,\cP)} < \infty$ and it follows that also $\E{RV(X-Y,\cP)} < \infty$. Thus Theorem~\ref{Thm:reverse_martingale} applies and we obtain that
 \begin{equation}\label{Eq:QV_symmetrized}
\Econd{RV(X-Y,\cP)}{\cG_\infty} = [X-Y,X-Y]_T = [X,X]_T + [Y,Y]_T.
 \end{equation}
Taking expectations and combining \eqref{Eq:QV_symmetrized} with \eqref{Eq:RV_decomp3} yields
\begin{align*}
 \Econd{[X,X]_T}{\cH} + \Econd{[Y,Y]_T}{\cH} &= \Econd{RV(X,\cP)}{\cH} + \Econd{RV(Y,\cP)}{\cH} - \\
 &- 2 \sum_{k=0}^{n-1}\Econd{X_{t_{k+1}} - X_{t_k}}{\cH} \cdot \Econd{Y_{t_{k+1}} - Y_{t_k}}{\cH}.
\end{align*}
Since $X$ and $Y$ have equal $\cH$-conditional distributions \eqref{Eq:RV_expectation} follows. It is obvious that the sum in \eqref{Eq:RV_expectation} vanishes if $X$ is a martingale. Conversely, if the sum vanishes for any partition $\cP$, then $\Econd{X_t - X_s}{\cH} = 0$ for any $0 \le s \le t$. By the $\cH$-conditional independence of increments, we conclude that for any $Z \in \cF_s$
\[\E{(X_t - X_s)Z} = \E{\Econd{X_t - X_s}{\cH}Z} = \Econd{X_t - X_s}{\cH} \cdot \Econd{Z}{\cH} = 0.\]
Hence $\Econd{X_t - X_s}{\cF_s} = 0$ and $X$ is a martingale.
\end{proof}

\subsection{Results for semimartingales with symmetric jump measure}

Theorem~\ref{Thm:reverse_martingale} on processes with symmetric and conditionally independent increments is typically not suitable for the applications to variance options that we have in mind. The reason is that the risk-neutral log-price $X$ of some asset can be a martingale only in degenerate cases, as already the asset price $S_t = S_0 e^{X_t}$ itself must be a martingale (at least locally). Assume for illustration that the risk-neutral security price $S_t = S_0 e^{X_t}$ is a true martingale and non-deterministic, then it follows immediately by Jensen's inequality that $\E{X_t} < \log \E{S_t/S_0} = X_0$ and hence that $X$ is not a martingale. Even if $S$ is a (non-deterministic) strictly local martingale that is bounded away from zero, it follows by Fatou's lemma and Jensen's inequality that
\[\E{X_t} \le \liminf_{n \to \infty} \E{X_{t \wedge \tau_n}} < \log \E{S_{t \wedge \tau_n}/S_0} = X_0,\]
for some localizing sequence $(\tau_n)_{n \in \NN}$, and hence that $X$ is not a martingale. For this reason we want to weaken the symmetry assumptions on $X$. We introduce the following assumptions:

\begin{ass}\label{myass}
The process $X$ is a square-integrable semimartingale starting at $X_0 = 0$ with the following properties:
\begin{enumerate}[(a)]
 \item $X$ has $\cH$-conditionally independent increments\label{Ass:PII};
  \item $X$ has no predictable times of discontinuity;
  \item $\nu$, the predictable compensator of the jump measure of $X$, is symmetric, i.e. it satisfies $\nu(\omega,dt,dx) = \nu(\omega,dt,-dx)$ a.s. \label{Ass:symmetry}.
\end{enumerate}
\end{ass}

 \begin{rem}`No predictable times of discontinuity' means that $\Delta X_\tau = 0$ a.s. for each predictable time $\tau$. This condition is also called \emph{quasi-left-continuity} of $X$. It does not rule out discontinuities at inaccessible stopping times, like the jumps of L\'evy processes. For processes with unconditionally independent increments, `no predictable times of discontinuity' can be replaced by `no fixed times of discontinuity' (see \citep[Cor.II.5.12]{Jacod1987}).
  \end{rem}

Let us briefly summarize some of the consequences of these assumptions. First, since $X$ is a square-integrable semimartingale, it is also a special semimartingale (cf. \citep[II.2.27, II.2.28]{Jacod1987}), and we can choose $h(x) = x$ as a truncation function, relative to which the semimartingale characteristics $(B,C,\nu)_h$ are defined. For this choice of truncation functions we have that $B = A$, i.e. the drift $B$ is exactly the predictable finite variation part of the canonical semimartingale decomposition. Second, a semimartingale has $\cH$-conditionally independent increments, if and only if there exists a version of its characteristics $(B,C,\nu)$ that is $\cH$-measurable. The increments are unconditionally independent, if and only if $\cH$ is trivial or equivalently if there exists a version of the characteristics $(B,C,\nu)$ that is deterministic (cf. \citep[Thm.~II.4.15]{Jacod1987}). Third, the assumption that $X$ has no predictable times of discontinuity implies that the version of the characteristics might be chosen such that they are continuous (cf. \citep[Prop.~II.2.9]{Jacod1987}).

 \begin{thm}\label{Thm:QV_order} Let $X$ be a process satisfying Assumption~\ref{myass}, and let $\cP$ be a partition of $[0,T]$. Then there exists a $\sigma$-algebra $\cG_\infty \subset \cF$ such that
  \begin{equation}\label{Eq:expectation_tail}
   \Econd{RV(X,\cP)}{\cG_\infty} = [X,X]_T + RV(B,\cP).
  \end{equation}
  If the increments of $X$ are unconditionally independent, it also holds that
  \begin{equation}\label{Eq:f_centering_tail}
   \Econd{\overline{RV}^h(X,\cP)}{\cG_\infty} \ge \overline{[X,X]}^h_T,
  \end{equation}
  for each $h \in \mathrm{Lip}_1(\Rplus)$, and with equality for $h(x) = x$.
 \end{thm}
 \begin{rem}
  In the proof below $\cG_\infty$ will be constructed explicitly as the tail $\sigma$-algebra of a sequence similar to \eqref{Eq:define_Gn}.
  \end{rem}
  \begin{proof}
 Let $(\cP^n)$ be a nested sequence of partitions with $\mathrm{mesh}(\cP^n) \to 0$, such that $\cP^N = \cP$ for some $N$. Let $X = Y + B$ be the canonical decomposition of $X$ into local martingale part and predictable finite variation part. Since $X$ is square-integrable $Y$ is a square-integrable martingale and $B$ has square-integrable total variation. In addition, define
 \begin{equation}\label{Eq:Gn_repeat}
 \cG_n = \cH \vee \sigma\Big(RV(Y,\cP^{n}), RV(Y,\cP^{n+1}), \dotsc\Big)
 \end{equation}
and let $\cG_\infty$ be the tail $\sigma$-algebra $\bigcap_{n \in \NN} \cG_n$ of $(\cG_n)_{n \in \NN}$.
Following \citep[Thm.~II.6.6]{Jacod1987}, the conditional characteristic function of an increment of $Y = X - B$ is given by
\begin{equation}\label{Eq:cond_char}
\begin{split}
 \phi_{t,s}(u) &:= \Econd{e^{iu(Y_t - Y_s)}}{\cH} = \\
 &= \exp \left(-\frac{u^2}{2}(C_t - C_s) + \int_s^t \int_{\RR}{\left(e^{iux} - 1 - iux\right)}\nu(\omega,dr,dx) \right).
\end{split}
\end{equation}
The symmetry of $\nu$ implies that $\phi_{t,s}(u) = \phi_{t,s}(-u)$, which by equation~\eqref{eq:charf_equal} implies the $\cH$-conditional symmetry of $Y_t - Y_s$. Hence $Y$ is a process with $\cH$-conditionally symmetric and independent increments. We also have that $\E{RV(Y,\cP^n)} < \infty$ since $Y$ is square-integrable. Thus Theorem~\ref{Thm:reverse_martingale} can be applied to $Y$ and we obtain that
\begin{equation}\label{Eq:tail_sigma1}
\Econd{RV(Y,\cP^n)}{\cG_\infty} = [Y,Y]_T,
\end{equation}
for all $n \in \NN$ and in particular for $n = N$ and $\cP^n = \cP$. Since $X$ is quasi-left-continuous, its predictable finite variation component $B$ is even continuous, and it follows by \citep[Prop.~I.4.49]{Jacod1987} that
\begin{equation}\label{Eq:QV_decomp}
[X,X]_T = [Y,Y]_T + [B,B]_T = [Y,Y]_T.
\end{equation}
Furthermore we can decompose the realized variance of $X$ as
\begin{equation}\label{Eq:RV_decomp_next}
 RV(X,\cP) = RV(Y,\cP) + RV(B,\cP) + \sum_{k=1}^n \left(Y_{t_k} - Y_{t_{k-1}}\right) \cdot \left(B_{t_k} - B_{t_{k-1}}\right).
\end{equation}
We set $U := \Econd{\left(Y_{t_k} - Y_{t_{k-1}}\right) \cdot \left(B_{t_k} - B_{t_{k-1}}\right)}{\cG_n}$ and show that $U = 0$, similar to the proof of Theorem~\ref{Thm:reverse_martingale}. We set $\cR_n := \sigma\Big(RV(Y,\cP^{n}), RV(Y,\cP^{n+1}), \dotsc \Big)$ and using the reflection principle from Lemma~\ref{Lem:symmetrization} we obtain for all $H \in \cH$ and $Q \in \cR_n$  that
\begin{align*}
\E{UHQ} &= \E{\Econd{\left(Y_{t_k} - Y_{t_{k-1}}\right) Q}{\cH} \left(B_{t_k} - B_{t_{k-1}}\right) H} = \\
&= \E{\Econd{\left(Y_{t_{k-1}} - Y_{t_k}\right) Q}{\cH} \left(B_{t_k} - B_{t_{k-1}}\right) H} = -\E{UHQ}
\end{align*}
and hence that indeed, $U=0$. Taking $\cG_\infty$-conditional expectations, we have
\[\Econd{\left(Y_{t_k} - Y_{t_{k-1}}\right) \cdot \left(B_{t_k} - B_{t_{k-1}}\right)}{\cG_\infty} = 0, \]
and thus
\begin{equation}\label{Eq:RV_decomp_next2}
 \Econd{RV(X,\cP)}{\cG_\infty} = \Econd{RV(Y,\cP)}{\cG_\infty} + RV(B,\cP).
\end{equation}
Combining \eqref{Eq:tail_sigma1}, \eqref{Eq:QV_decomp} and \eqref{Eq:RV_decomp_next2} yields \eqref{Eq:expectation_tail}. Suppose now that the increments of $X$ are unconditionally independent. Then $B$ is even deterministic (see the discussion after Assumption~\ref{myass}), and we derive from \eqref{Eq:expectation_tail} that
\[\Econd{\overline{RV}^f(X,\cP)}{\cG_\infty} = [X,X]_T + RV(B,\cP) - f\Big(\E{[X,X]_T} + RV(B,\cP)\Big).\]
Applying the Lipschitz property of $f$ \eqref{Eq:f_centering_tail} follows.
\end{proof}

\subsection{Convex Order Relations}

By a simple application of Jensen's inequality, Theorem~\ref{Thm:QV_order} can be translated into a convex order relation between the realized variance and the quadratic variation of $X$.

\begin{thm}\label{Thm:order1} Let $X$ be a process satisfying Assumption~\ref{myass}, and let $\cP$ be a partition of $[0,T]$. Then the following holds:
\begin{enumerate}[(a)]
\item For any increasing convex function $f$
\begin{equation}\label{Eq:convex_order4}
\E{f\Big(RV(X,\cP)\Big)} \ge \E{f\Big([X,X]_T\Big)}.
\end{equation}
If $X$ is a martingale, then \eqref{Eq:convex_order4} holds for any convex function $f$.
\item If the increments of $X$ are unconditionally independent, then the $h$-centered realized variance satisfies
\begin{equation}\label{Eq:convex_order5}
\E{f\Big(\overline{RV}^h(X,\cP)\Big)} \ge \E{f\Big(\overline{[X,X]}_T^h\Big)}
\end{equation}
for any increasing convex function $f$ and $h \in \mathrm{Lip}_1(\Rplus)$. If $X$ is a martingale or $h(x) = x$, then \eqref{Eq:convex_order5} holds for any convex function $f$.
\end{enumerate}
\end{thm}

Summing up the theorem, we can say that the `increasing convex order conjecture' set forth in the introduction holds true at least for square-integrable semimartingales $X$ with conditionally independent increments, symmetric jumps and without predictable times of discontinuity. If $X$ is also a martingale, then even the `convex order conjecture' holds true. However, as discussed at the beginning of the section, $X$ cannot be a martingale if $e^X$ is assumed to be a non-deterministic martingale, such that only the increasing convex order conjecture holds in the relevant applications in finance.

\begin{proof} Applying an increasing convex function $f$ to both sides of \eqref{Eq:expectation_tail} and using Jensen's inequality, we have that
\[\Econd{f\Big(RV(X,\cP)\Big)}{\cG_\infty} \ge f\Big([X,X]_T\Big).\]
Taking (unconditional) expectations the convex order relation \eqref{Eq:convex_order4} follows for all increasing convex $f$. If $X$ is a martingale, then by Corollary~\ref{cor:expectations} $\E{RV(X,\cP)} = \E{[X,X]_T}$ and it follows from \citep[Thm.4.A.35]{Shaked2007} that \eqref{Eq:convex_order4} holds even for any convex $f$. Equation \eqref{Eq:convex_order5} is derived analogously from \eqref{Eq:f_centering_tail}. In the case that $X$ is a martingale or $h(x) = x$ `increasing convex' can again be replaced by `convex' since $\E{\overline{RV}^h(X,\cP)} = \E{\overline{[X,X]}^h_T}$.
\end{proof}

\section{Results on predictable quadratic variation}

At this point we want to bring the \emph{predictable quadratic variation} $\scal{X}{X}$ of $X$ into the picture. As discussed in the introduction we conjecture the convex order relation $[X,X]_T \ge_\text{cx} \scal{X}{X}_T$ for many processes of interest. Indeed the following holds.

\begin{thm}\label{Thm:predictable} Let $X$ be a square-integrable semimartingale with $\cH$-conditionally independent increments. Then the quadratic variation $[X,X]$ dominates the predictable quadratic variation $\scal{X}{X}$ in convex order, that is 
\begin{equation}\label{Eq:convex_order6}
\E{f\Big([X,X]_T\Big)} \ge \E{f\Big(\scal{X}{X}_T\Big)}.
\end{equation}
for any convex function $f$.
\end{thm}
\begin{proof}
Let $(B,C,\nu)$ be the characteristics of the semimartingale $X$ relative to the truncation function $h(x) = x$. Following \citet[Sec.~II.2.a]{Jacod1987} the quadratic variation of $X$ is given by
\[[X,X]_t = C_t^2 + \int_0^t \int_\RR x^2 \mu(\omega,dx,ds),\]
where $\mu(\omega,dx,ds)$ is the random measure associated to the jumps of $X$. The predictable quadratic variation in turn is given (cf. \citet[Prop.~II.2.29]{Jacod1987}) by 
\[\scal{X}{X}_t = C_t^2 + \int_0^t \int_\RR x^2 \nu(\omega,dx,ds).\]
Furthermore, as $X$ has $\cH$-conditionally independent increments, the characteristics $(B,C,\nu)$ of $X$ are $\cH$-measurable. The process $[X,X]_t - \scal{X}{X}_t$ is a local $\FF$-martingale, and due to the square-integrability assumption on $X$ a true $\FF$-martingale. Together with the fact that $\cH \subset \cF_0$ we conclude that
\[\Econd{[X,X]_T - \scal{X}{X}_T}{\cH} = \Econd{\Econd{[X,X]_T - \scal{X}{X}_T}{\cF_0}}{\cH} = 0.\]
Applying a convex function $f$ to both sides of $\Econd{[X,X]_T}{\cH} = \scal{X}{X}_T$ and using Jensen's inequality we obtain
\[\Econd{f\Big([X,X]_T\Big)}{\cH} \ge f\Big(\scal{X}{X}_T\Big).\]
The desired result follows by taking expectations of this inequality.
\end{proof}

\section{Applications to variance options}
We now apply the results of Section~\ref{Sec:main} to options on variance. We distinguish between models with unconditionally independent increments, such as exponential L\'evy models and Sato processes, and models with conditionally independent increments, such as stochastic volatility models with jumps and time-changed L\'evy models. In the first class of models we obtain stronger results than in the latter.

\subsection{Models with independent increments}
\subsubsection{Exponential L\'evy models with symmetric jumps}\label{Sub:levy}
In an exponential L\'evy model the log-price $X$ is given by a L\'evy process with L\'evy triplet $(b,\sigma^2 ,\nu)$. Clearly $X$ is a semimartingale with independent increments, and without fixed times of discontinuity. If the L\'evy measure $\nu$ is symmetric and satisfies $\int x^2 \nu(dx) < \infty$ then Assumption~\ref{myass} holds with trivial $\cH$ and the results of Section~\ref{Sec:main} apply. We list some commonly used L\'evy models with symmetric jump measure; for definitions and notation we refer to \citet[Ch.~4]{Cont2004}, or in case of the CGMY model to \citet{Carr2003a}:
\begin{itemize}
\item the Black-Scholes model,
\item the Merton model with $0$ mean jump size,
\item the Kou model with symmetric tail decay and equal tail weights ($\lambda_+ = \lambda_-$ and $p = 1/2$),
\item the Variance Gamma process with $\theta = 0$,
\item the Normal Inverse Gaussian process with $\theta = 0$,
\item the CGMY model with symmetric tail decay ($G = M$),
\item the generalized hyperbolic model with symmetric tail decay ($\beta = 0$).
\end{itemize}

\subsubsection{Symmetric Sato processes}\label{Sub:sato}
A Sato process is a stochastically continuous process that has independent increments and is self-similar in the sense that $X_{\lambda t} \stackrel{d}{=} \lambda^\gamma X_t$ for all $\lambda, t \ge 0$ and for some fixed exponent $\gamma > 0$. Such processes have been studied systematically by \citet{Sato1991}, and used for financial modeling for example in \citet{Carr2007}. In \citet{Carr2010} the authors consider options on quadratic variation, when the log-price process $X$ is a Sato process. Following \citep{Sato1991}, a Sato process is a semimartingale with the deterministic characteristics $(t^\gamma b, t^{2\gamma} \sigma^2, \nu_t(dx))$ where $\nu_t(dx)$  is of the form $\nu_t(B) = \int{\Ind{B}(t^\gamma x)}\nu(dx)$ for some fixed L\'evy measure $\nu$. If $\nu$ is symmetric with $\int x^2 \nu(dx) < \infty$ then $X$ satisfies Assumption~\ref{myass} with trivial $\cH$. Referring to the terminology and notation of \citep{Carr2007}, this symmetry condition is satisfied for the VGSSD model with $\theta = 0$, the NIGSSD model with $\theta = 0$ and the MXNRSSD model with $b = 0$.

\subsubsection{Order relations between option prices}
For the models described in \ref{Sub:levy} and \ref{Sub:sato}, Theorem~\ref{Thm:order1} implies the following order relations between the prices of options on variance:

\begin{description}
\item[Fixed Strike Calls] Consider a call on variance with a fixed strike, i.e. with a payoff function $x \mapsto (x - K)^+$ where $K \ge 0$. In this case Theorem~\ref{Thm:order1} applies with $h = 0$ and $f(x) = (x - K)^+$. Hence the price of a fixed-strike call on discretely observed realized variance is larger than the price of its continuous counterpart.

 \item[Relative Strike Calls] Let $s$ be the swap rate, and consider a call with payoff function $x \mapsto (ks - x)^+$. The difference to the fixed strike case is that now the swap rate $s$ also depends on whether discrete or continuous observations are used for pricing. If $k \le 1$, that is if the call is in-the-money or at-the-money then Theorem~\ref{Thm:order1} applies with $h = kx$ and $f(x) = (x)^+$. Hence, the price of a relative-strike in-the-money-call on discrete realized variance is larger than the price of its continuous counterpart.

 \item[ATM Puts and Straddles] Consider an at-the-money put with payoff $x \mapsto (s - x)^+$ or an at-the-money straddle with payoff $x \mapsto |s - x|$, where $s$ is the swap rate. In these cases Theorem~\ref{Thm:QV_order} applies with $h(x) = x$ and $f(x) = (-x)^+$ or $f(x) = |x|$. Hence, the price of an at-the-money put or an at-the-money straddle on discrete realized variance is larger than the price of its continuous counterpart.
\end{description}

\subsection{Models with conditionally independent increments}
\subsubsection{Time-changed L\'evy processes with symmetric jumps}\label{Sub:tc_levy}
Consider a L\'evy process $L$ with L\'evy triplet $(b,\sigma,\nu)$ where the jump measure $\nu$ is symmetric with $\int x^2 \nu(dx) < \infty$, and a continuous, integrable and increasing process $\tau_t$, independent of $L$. Letting $\tau_t$ act as a stochastic clock, we may define the time-changed process $X_t = L_{\tau_t}$. The process $X$ is a semimartingale with characteristics $(b\tau_t,\sigma \tau_t, \nu \tau_t)$. Defining $\cH = \sigma(\tau_t, t \ge 0)$ and $\cF_t = \cH \vee \sigma(X_s, 0 \le s \le t)$ we observe that $X$ is a $\FF$-adapted square-integrable semimartingale with $\cH$-conditionally independent increments and symmetric jumps, and hence satisfies Assumption~\ref{myass}.

\subsubsection{Stochastic volatility models with jumps}\label{Sub:sv_jump}
Suppose that the risk-neutral log-price $X$ has the semimartingale representation
\begin{equation}\label{Eq:generic_SV}
dX_t = b_t \,dt + \sigma_t\,dW_t + x \star \left(\mu(dt,dx) - \nu_t(dx) dt\right), \qquad X_0 = 0
\end{equation}
where $W$ is a $\FF$-Brownian motion and $\mu(dt,dx)$ a $\FF$-Poisson random measure with compensator $\nu_t(x) dt$, and $b_t, \sigma_t, \nu_t$ are $\cH$-measurable with $\cH \subset \cF_0$. Moreover, assume that $\nu_t(dx)$ is symmetric and that
\[\int_0^T{\left(\E{b_t^2} + \E{\sigma_t^2} + \E{\int x^2 \nu_t(dx)}\right)dt} < \infty,\]
such that $X$ is square-integrable (cf. \citet[p.~8f]{KM2010}). Then $X$ satisfies Assumption~\ref{myass} and the results of Section~\ref{Sec:main} apply. This setting includes the uncorrelated Heston model and extensions with symmetric jumps.

\subsubsection{Order relations between option prices}
For the models described in \ref{Sub:tc_levy} and \ref{Sub:sv_jump}, Theorem~\ref{Thm:order1} implies the following order relation between the prices of options on variance:

\begin{description}
\item[Fixed Strike Calls] Consider a call on variance with a fixed strike, i.e. with a payoff function $x \mapsto (x - K)^+$ where $K \ge 0$. In this case Theorem~\ref{Thm:order1} applies with $h = 0$ and $f(x) = (x - K)^+$. Hence the price of a fixed-strike call on discretely observed realized variance is always bigger than the price of its continuous counterpart.
\end{description}

\section{Conclusions and counterexamples}
In this article we have shown that the increasing convex order conjecture for realized variance holds true in a large class of asset pricing models, namely under the assumption that the log-price $X$ is a square-integrable semimartingale with conditionally independent increments, symmetric jumps and no predictable times of discontinuity. However, the numerical evidence given e.g. in \citet{KM2010} suggests that the conjecture may hold true under even more general conditions. In particular for L\'evy models with asymmetric jumps the conjecture seems to be valid, although this case is not covered by the assumptions of the article. It is not obvious -- at least not to the authors -- how the strategy of the proof can be extended to these cases, in particular how the symmetry condition on the jumps can be relaxed. On the other hand removing the assumption of conditionally independent increments easily leads to a violation of the convex order conjecture. A numerical counterexample for a stochastic volatility model with nonzero leverage parameter has recently appeared in \citet{Bernard2012}. Another counterexample can be given using a time-changed Brownian motion; it is based on suggestions by Walter Schachermayer and David Hobson.
\begin{counterexample}Let $B$ be an $\FF$-Brownian motion and define the stopping time $\tau = \inf \set{t: B_t \not \in [-1,1]}$. The stopped process $B^\tau$ is a bounded martingale and can be closed at infinity by adding the terminal value $B^\tau_\infty = B_\tau$, which only takes the values $\pm 1$. Let $f$ be an increasing, continuous and bijective function from $[0,1]$ to $[0,\infty]$ (such as $t \mapsto \tan \left(t\tfrac{\pi}{2}\right)$). Define on the interval $[0,1]$ the process $X = \left((B^\tau)_{f(t)}\right)_{t \in [0,1]}$ which is a continuous martingale w.r.t $\FF' = \left(\cF_{f(t)}\right)_{t \in [0,1]}$. Using the continuity of $f$ we obtain $[X,X]_1 = \lim_{t \to 1} [B^\tau,B^\tau]_{f(t)} = \tau$. Choosing the partition $\cP = \set{0,1}$ we get $RV(X,\cP) = (\pm 1)^2 = 1$ a.s. Using Jensen's inequality and the fact that $\E{\tau} = 1$ we conclude that $[X,X]_1 \ge_\text{cx} RV(X,\cP)$. Since $[X,X]_1$ is truly random the opposite relation $RV(X,\cP) \ge_\text{cx} [X,X]_1$ cannot hold true. This counterexample can easily be extended to other partitions of $[0,1]$ by introducing additional stopping times.
\end{counterexample}
The next counterexample shows that also the convex order relation $[X,X]_{T} \geq_{\text{cx}} \langle X,X \rangle _{T}$ between predictable and ordinary quadratic variation shown in Theorem~\ref{Thm:predictable} collapses if the assumption of conditionally independent increments is violated.

\begin{counterexample}
Let $\theta$ be a random variable taking the values $\pm 1$ with probability $\frac{1}{2}$ each. In addition, let $\tau$ be a uniformly distributed random time on $[0,1]$. Define the cadlag process $Z_t := \theta \Ind{\tau \le t}$. This process is zero up to time $\tau$ and then jumps to one of the values $\pm 1$. Denote by $\HH$ the natural filtration of $Z$; it is obvious that $Z$ is a $\HH$-martingale. By direct calculation
\[[Z,Z]_t = \theta^2 \Ind{\tau \le t} = \Ind{\tau \le t}.\]
In particular $[Z,Z]_1 = 1$ almost surely. By \citet{Bielecki2002} the `hazard function' $H$ corresponding to the random time $\tau$ is given by $H(t) = -\log(1-t)$ and $\Ind{\tau \le t} - H(t \wedge \tau)$ is a martingale. We conclude that $t \mapsto H(t \wedge \tau)$ is the predictable compensator of $t \mapsto \Ind{\tau \le t}$ and hence that
\[\scal{Z}{Z}_t = -\log(1 - (t \wedge \tau)).\]
In particular $\scal{Z}{Z}_1 = -\log(1 - \tau)$, i.e. it is exponentially distributed with rate $1$, and hence $\scal{Z}{Z}_1 \ge_\text{cx} [Z,Z]_1$, while the reverse relation does not hold true.
\end{counterexample}

\bibliographystyle{plainnat}
\bibliography{convex_ref}

\begin{thebibliography}{28}
\providecommand{\natexlab}[1]{#1}
\providecommand{\url}[1]{\texttt{#1}}
\expandafter\ifx\csname urlstyle\endcsname\relax
  \providecommand{\doi}[1]{doi: #1}\else
  \providecommand{\doi}{doi: \begingroup \urlstyle{rm}\Url}\fi

\bibitem[Bernard and Cui(2012)]{Bernard2012}
C.~Bernard and Zhenyu Cui.
\newblock Prices and asymptotics for discrete variance swaps.
\newblock Preprint, 2012.

\bibitem[Bielecki and Rutkowski(2002)]{Bielecki2002}
T.~Bielecki and M.~Rutkowski.
\newblock \emph{Credit Risk: Modeling, Valuation and Hedging}.
\newblock Springer, 2002.

\bibitem[Broadie and Jain(2008{\natexlab{a}})]{Broadie2008a}
M.~Broadie and Jain.
\newblock Pricing and hedging volatility derivatives.
\newblock \emph{The Journal of Derivatives}, 15\penalty0 (3):\penalty0 7--24,
  2008{\natexlab{a}}.

\bibitem[Broadie and Jain(2008{\natexlab{b}})]{Broadie2008}
M.~Broadie and A.~Jain.
\newblock The effect of jumps and discrete sampling on volatility and variance
  swaps.
\newblock \emph{International Journal of Theoretical and Applied Finance},
  11:\penalty0 761--797, 2008{\natexlab{b}}.

\bibitem[B{\"u}hler(2006)]{Buhler2006}
Hans B{\"u}hler.
\newblock \emph{Volatility Markets -- Consistent modeling, hedging and
  practical implementation}.
\newblock PhD thesis, TU Berlin, 2006.

\bibitem[Carr and Lee(2009)]{Carr2009}
P.~Carr and Lee.
\newblock Volatility derivatives.
\newblock \emph{Annual Review of Financial Economics}, 1:\penalty0 1--21, 2009.

\bibitem[Carr et~al.(2003)Carr, Geman, Madan, and Yor]{Carr2003a}
P.~Carr, H.~Geman, D.~Madan, and M.~Yor.
\newblock Stochastic volatility for levy processes.
\newblock \emph{Mathematical Finance}, 13\penalty0 (3):\penalty0 345--382,
  2003.

\bibitem[Carr et~al.(2007)Carr, Geman, Madan, and Yor]{Carr2007}
P.~Carr, H.~Geman, D.~Madan, and M.~Yor.
\newblock Self-decomposability and option pricing.
\newblock \emph{Mathematical Finance}, 17\penalty0 (1):\penalty0 31--57, 2007.

\bibitem[Carr and Lee(2008)]{Carr2003}
Peter Carr and Roger Lee.
\newblock Robust replication of volatility derivatives.
\newblock NYU Courant Institute Mathematics in Finance Working Paper {\#}2008 -
  3, 2008.

\bibitem[Carr et~al.(2010)Carr, Geman, Madan, and Yor]{Carr2010}
Peter Carr, Helyette Geman, Dilip~B. Madan, and Marc Yor.
\newblock Options on realized variance and convex orders.
\newblock \emph{Quantitative Finance}, iFirst:\penalty0 1--10, 2010.

\bibitem[Cogburn and Tucker(1961)]{Cogburn1961}
Robert Cogburn and Howard~G. Tucker.
\newblock A limit theorem for a function of the increments of a decomposable
  process.
\newblock \emph{Transactions of the American Mathematical Society},
  99:\penalty0 278--284, 1961.

\bibitem[Cont and Tankov(2004)]{Cont2004}
Rama Cont and Peter Tankov.
\newblock \emph{Financial Modelling with Jump Processes}.
\newblock Financial Mathematics Series. Chapman {\&} Hall/CRC, 2004.

\bibitem[Farkas and Drimus(2010)]{Farkas2010}
Walter Farkas and Gabriel Drimus.
\newblock Options on discretely sampled variance: Discretization effect and
  greeks.
\newblock Working Paper, 2010.

\bibitem[Gatheral(2008)]{Gatheral2008}
Jim Gatheral.
\newblock Consistent modeling of spx and vix options.
\newblock Presentation at the 5th Bachelier Congress, London, 2008.

\bibitem[Jacod and Shiryaev(1987)]{Jacod1987}
J.~Jacod and A.N. Shiryaev.
\newblock \emph{Limit Theorems for Stochastic Processes}.
\newblock Springer, 1987.

\bibitem[Kallenberg(1997)]{Kallenberg1997}
Olav Kallenberg.
\newblock \emph{Foundations of Modern Probability}.
\newblock Springer, 1997.

\bibitem[Kallsen et~al.(2009)Kallsen, Muhle-Karbe, and Voss]{Kallsen2009}
Jan Kallsen, Johannes Muhle-Karbe, and Moritz Voss.
\newblock Pricing options on variance in affine stochastic volatility models.
\newblock to appear in Mathematical Finance, 2009.

\bibitem[Keller-Ressel and Muhle-Karbe(2010)]{KM2010}
Martin Keller-Ressel and Johannes Muhle-Karbe.
\newblock Asymptotic behaviour and pricing of options on variance.
\newblock forthcoming in Finance \& Stochastics, arXiv:1003.5514v3, 2010.

\bibitem[Lee(2010)]{Lee2010a}
Lee.
\newblock Weighted variance swap.
\newblock In Rama Cont, editor, \emph{Encyclopedia of {Q}uantitative
  {F}inance}. Wiley, 2010.

\bibitem[L{\'e}vy(1940)]{Levy1940}
Paul L{\'e}vy.
\newblock Le mouvement brownien plan.
\newblock \emph{American Journal of Mathematics}, 62\penalty0 (1):\penalty0 487
  -- 550, 1940.

\bibitem[Lindvall(1992)]{Lindvall1992}
Torgny Lindvall.
\newblock \emph{Lectures on the coupling method}.
\newblock John Wiley {\&} Sons, New York, 1992.

\bibitem[Lo{\`e}ve(1963)]{Loeve1963}
Michel Lo{\`e}ve.
\newblock \emph{Probability Theory}.
\newblock Van Nostrand Reinhold, 3rd edition, 1963.

\bibitem[Neuberger(1992)]{Neuberger1992}
A.~Neuberger.
\newblock Volatility trading.
\newblock London Business School Working Paper, 1992.

\bibitem[Protter(2004)]{Protter2004}
Philip~E. Protter.
\newblock \emph{Stochastic Integration and Differential Equations}.
\newblock Springer, 2004.

\bibitem[Sato(1991)]{Sato1991}
Ken-Iti Sato.
\newblock Self-similar processes with independent increments.
\newblock \emph{Probability Theory and Related Fields}, 89:\penalty0 285--300,
  1991.

\bibitem[Sato(1999)]{Sato1999}
Ken-Iti Sato.
\newblock \emph{{L}{\'e}vy processes and infinitely divisible distributions}.
\newblock Cambridge University Press, 1999.

\bibitem[Sepp(2010)]{Sepp2010}
A.~Sepp.
\newblock Note on `pricing options on realized variance in the {H}eston model
  with jumps in returns and volaility'.
\newblock Working Paper, 2010.

\bibitem[Shaked and Shanthikumar(2007)]{Shaked2007}
Moshe Shaked and J.~George Shanthikumar.
\newblock \emph{Stochastic Orders}.
\newblock Springer, New York, 2007.

\end{thebibliography}
\end{document}